\documentclass[a4paper,UKenglish]{article}
\usepackage{amssymb}
\usepackage{amsthm}
\usepackage{amsmath}
\usepackage{graphicx}
\usepackage[utf8]{inputenc}
\usepackage{vmargin, fancyhdr} 
\usepackage[hidelinks]{hyperref}

\usepackage{tikz}
\usetikzlibrary{patterns}
\usetikzlibrary{decorations,arrows}
\usetikzlibrary{decorations.pathmorphing}
\usepgflibrary{decorations.pathreplacing} 
\usetikzlibrary{positioning}

\newtheorem{lemma}{Lemma}
\newtheorem{theorem}{Theorem}

\theoremstyle{plain}
\newtheorem*{claim}{Claim}

\newcommand{\scalePart}{4z(7z + 1)}
\newcommand{\jobs}{J}
\newcommand{\sched}{\sigma}

\title{Complexity and Inapproximability Results for Parallel Task Scheduling and Strip Packing\thanks{This work was partially supported by German Research Foundation (DFG) project JA 612 /14-2.}}
\author{Sören Henning, Klaus Jansen, Malin Rau, Lars Schmarje\\
Institut für Informatik, Christian-Albrechts-Universität zu Kiel, Germany\\
\texttt{$\{$stu114708,kj,mra,stu115194$\}$@informatik.uni-kiel.d}e}

\begin{document}

\maketitle

\begin{abstract}
We study the Parallel Task Scheduling problem $Pm|size_j|C_{\max}$ with a constant number of machines.
This problem is known to be strongly NP-complete for each $m \geq 5$, while it is solvable in pseudo-polynomial time for each $m \leq 3$.
We give a positive answer to the long-standing open question whether this problem is strongly $NP$-complete for $m=4$. 
As a second result, we improve the lower bound of $\frac{12}{11}$ for approximating pseudo-polynomial Strip Packing to $\frac{5}{4}$. Since the best known approximation algorithm for this problem has a ratio of $\frac{4}{3} + \varepsilon$, this result narrows the gap between approximation ratio and inapproximability result by a significant step.
Both results are proven by a reduction from the strongly $NP$-complete problem 3-Partition.

%We study the Parallel Task Scheduling problem with $4$ machines $P4|size_j|C_{\max}$ and give a positive answer to the long-standing open question if this problem is strongly $NP$-complete. Furthermore, we prove a lower bound of $\frac{5}{4}$ for the best possible approximation ratio of pseudo-polynomial Strip Packing, under the assumption that $P \not = NP$. Both results are proven by a reduction from the strongly $NP$-complete problem 3-Partition.  
\end{abstract}
\section{Introduction}
In the Parallel Task Scheduling problem, we have given $m$ machines and a set of jobs $\jobs$. Each job $j \in J$ has a processing time $p(j) \in \mathbb{N}$ and a number of required machines $q(j) \in \mathbb{N}$. A schedule $\sched$ is a combination of two functions $\sched: \jobs \rightarrow \mathbb{N}$ and $\rho: \jobs \rightarrow \{M | M \subseteq \{1,\dots, m\} \}$. The function $\sched$ maps each job to a start point in the schedule, while $\rho$ maps each job to the set of machines it is processed on. 
A schedule is feasible if each machine processes at most one job at the time and each job is processed on the required number of machines.
The objective is to find a feasible schedule $\sched$ minimising the makespan $T := \max_{i \in \jobs} \sched(i)+p(i)$. 
This problem is denoted with $P|size_j|C_{\max}$. If the number of machines is constant we write $Pm|size_j|C_{\max}$. For a given job $j \in \jobs$ we define its work as $w(j) := p(j) \cdot q(j)$. For a subset $\jobs' \subseteq \jobs$ we define its total work as $w(\jobs') := \sum_{j \in \jobs'} w(j)$.

In the Strip Packing problem we have given a strip with a width $W \in \mathbb{N}$ and infinite height as well as a set of rectangular items $I$. Each item $i \in I$ has a width $w_i \in \mathbb{N}_{\leq W}$ and a height $h_i \in \mathbb{N}$. 
The objective is to find a feasible packing of the items $I$ into the strip, which minimizes the packing height. 
A \textit{packing} of the items $I$ into the strip is a function $\rho: I \rightarrow \mathbb{Q}_0 \times \mathbb{Q}_{0}$, which assigns the left bottom corner of an item to a position in the strip, such that for each item $i \in I$ with  $\rho(i) = (x_i,y_i)$ we have $x_i + w_i \leq W$. 
An \textit{inner} point of $i$ is a point from the set $inn(i) := \{(x,y)\in \mathbb{R}\times \mathbb{R}| x_i < x < x_i +w_i, y_i < y < y_i + h_i \}$. 
We say two items $i,j \in I$ \textit{overlap} if they share an inner point (i.e if $inn(i) \cap inn(j) \not = \emptyset$).
A packing is \textit{feasible} if no two items overlap. The height of a packing is defined as $H := \max_{i \in I} y_i + h_i$. 

A well known and interesting fact is that, in this setting, we can transform feasible packings to packings where all positions are integral, without enlarging the packing height \cite{BansalCKS06}. This can be done by shifting all items downwards until they touch the upper border of an item or the bottom of the strip. Now all $y$-coordinates of the items are integral since each is given by the sum of some item heights, which are integral. 
The same can be done for the $x$-coordinate by shifting all items to the left as far as possible. Therefore we can assume that we have packings of the form $\rho: I \rightarrow \mathbb{N}_0 \times \mathbb{N}_{0}$.

Strip packing is closely related to Parallel Task Scheduling. If we demand that jobs from the Parallel Task Scheduling are processed on contiguous machines, the resulting problem is equivalent to the Strip Packing problem. Although these problems are closely related, there are some instances that have a smaller optimal value on non-contiguous machines than on contiguous machines \cite{TurekWY92}.

\paragraph*{Related Work}
\subparagraph*{Parallel Task Scheduling:}
In 1989 Du and Leung \cite{DuL89a} proved the problem $Pm|size_j|C_{\max}$ to be strongly $NP$-complete for all $m \geq 5$, while $Pm|size_j|C_{\max}$ is solvable by a pseudo-polynomial algorithm for all $m \leq 3$. 
%An algorithm is pseudo-polynomial if its running time depends polynomially on numeric values of the input, like the sum of all processing times, which occur just in logarithmic size in the input length.
Amoura et al. \cite{AmouraBKM02}, as well as Jansen and Porkolab \cite{JansenP02}, presented a polynomial time approximation scheme (in short PTAS) for the case that $m$ is a constant. A PTAS is a family of algorithms that finds a solution with an approximation ratio of $(1 +\varepsilon)$ for any given value $\varepsilon > 0$.
If $m$ is polynomially bounded by the number of jobs, a PTAS still exists \cite{JansenT10}. Nevertheless, if $m$ is arbitrarily large, the problem gets harder. By a simple reduction from the Partition problem, one can see that there is no polynomial algorithm with approximation ratio smaller than $\frac{3}{2}$. Furthermore, there is no asymptotic algorithm with approximation ratio $\alpha \mathrm{OPT} + \beta$, where $\alpha < 3/2$ and $\beta$ polynomial in $n$ \cite{Johannes06}.
Parallel Task Scheduling with arbitrarily large $m$ has been widely studied \cite{GareyG75,TurekWY92,LudwigT94,FeldmannST94}.
The algorithm with the best known absolute approximation ratio of $\frac{3}{2}+\varepsilon$ was presented by Jansen \cite{Jansen12}.

\subparagraph*{Strip Packing:}
The Strip Packing problem was first studied in 1980 by Baker et al. \cite{BakerCR80}. They presented an algorithm with an absolute approximation ratio of 3. 
This ratio was improved by a series of papers \cite{CoffmanGJT80,Sleator80,Schiermeyer94,Steinberg97,HarrenS09}. The algorithm with the best known absolute approximation ratio by Harren, Jansen, Prädel and van Stee \cite{harren20145} achieves a ratio of $\frac{5}{3} +\varepsilon$. By a simple reduction from the Partition problem, one can see that it is impossible to find an algorithm with better approximation ratio than $\frac{3}{2}$, unless $P = NP$. 

However, asymptotic algorithms can achieve asymptotic approximation ratios better than $\frac{3}{2}$ and have been studied in various papers \cite{CoffmanGJT80,Golan,BakerBK81}. Kenyon and Rémila \cite{Kenyon00} presented an asymptotic fully polynomial approximation scheme (in short AFPTAS) with additive term $\mathcal{O}(h_{\max}/\varepsilon^2)$, where $h_{max}$ is the largest occurring item height. An approximation scheme is fully polynomial if its running time is polynomial in $1/\varepsilon$ as well.
This algorithm was simultaneously improved by Sviridenko \cite{Sviridenko12} and Bougeret et al. \cite{BougeretDJRT2011} to an algorithm with an additive term of $\mathcal{O}(\log(1/\varepsilon)/\varepsilon)h_{\max}$. 
Furthermore, at the expense of the running time, Jansen and Soils-Oba \cite{JansenS09} presented an asymptotic PTAS with an additive term of $h_{\max}$.

Recently the focus shifted to pseudo-polynomial algorithms.
Jansen and Thöle \cite{JansenT10} presented an pseudo-polynomial algorithm with approximation ratio of $\frac{3}{2} + \varepsilon$. Since the partition problem is solvable in pseudo-polynomial time, the lower bound of $\frac{3}{2}$ for polynomial time Strip Packing can be beaten by pseudo-polynomial algorithms. The first such algorithm with a better approximation ratio than $\frac{3}{2}$ was given by Nadiradze and Wiese \cite{NadiradzeW16}. It has an absolute approximation ratio of $\frac{7}{5} + \varepsilon$. Its approximation ratio was independently improved to $\frac{4}{3} + \varepsilon$ by Galvez, Grandoni, Ingala, and Khan \cite{GalvezGIK16} and by Jansen and Rau \cite{JansenR16}. All these algorithms have a polynomial running time if the width of the strip $W$ is bounded by a polynomial in the number of items. 

In contrast to Parallel Task Scheduling, Strip Packing can not be approximated arbitrarily close to 1, if we allow pseudo-polynomial running time. Adamaszek, Kociumaka, Pilipczuk and Pilipczuk \cite{Adamaszek16} proved this by presenting a lower bound of $\frac{12}{11}$. This result also implies that Strip Packing admits no quasi-polynomial time approximation scheme, unless $NP \subseteq DTIME(2^{\mathrm{polylog}(n)})$.
Christensen, Khan, Pokutta, and Tetali \cite{Christensen2017} list 10 major open problems related to multidimensional Bin Packing. As the 10th problem they name pseudo polynomial Strip Packing and underline the importance of finding tighter pseudo-polynomial time results for lower and upper bounds.

\paragraph*{New Results}

In this paper, we present two hardness results. The first result answers the long-standing open question whether the problem $P4|size_j|C_{\max}$ is strongly $NP$-complete.  

\begin{theorem}
\label{thm:SchedulingHardness}
The Parallel Tasks Scheduling problem on $4$ machines $P4|size_j|C_{max}$ is strongly NP-complete.
\end{theorem}

The second result concerns pseudo-polynomial Strip Packing. We manage to adapt our reduction for $P4|size_j|C_{\max}$ to Strip Packing, by transforming the optimal schedule into a packing of rectangles interpreting the makespan as the width of the strip. 
%by interpreting the optimal makespan of a schedule for $P4|size_j|C_{\max}$ as the width of the strip. 
This adaptation leads to the following result:

\begin{theorem}
\label{thm:StripPackingHardness}
For each $\varepsilon >0$ it is NP-Hard to approximate Strip Packing with a ratio of $\frac{5}{4} - \varepsilon$ in pseudo-polynomial time.
\end{theorem}

This improves the so far best lower bound of $\frac{12}{11}$ to $\frac{5}{4}$. 
In Figure \ref{fig:improvement} we display the results for pseudo-polynomial Strip Packing achieved so far.

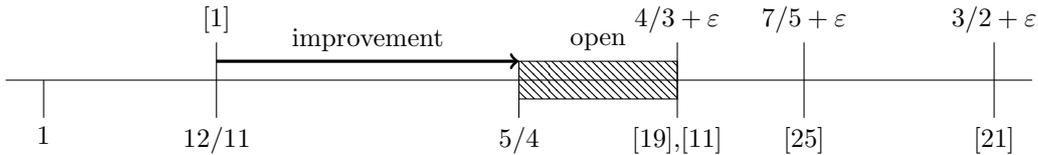
\begin{figure}
\begin{tikzpicture}
\pgfmathsetmacro{\w}{25}
\draw[-](-0.5 +\w,0)--(3/2 * \w + 0.5,0);
\draw(\w,0)--(\w,-0.5) node [below] {$1$};
\draw(12/11 * \w,0.5)node [above]{\cite{Adamaszek16}}--(12/11 * \w,-0.5) node [below] {$12/11$\tiny{}};
\draw(5/4 * \w,0)--(5/4 * \w,-0.5) node [below] {$5/4$};
\draw(4/3 * \w,- 0.5)node [below]{{\cite{JansenR16},\cite{GalvezGIK16}}} --(4/3 * \w,0.5) node [above] {$4/3+\varepsilon$};
\draw(7/5 * \w,-0.5)node [below]{{\cite{NadiradzeW16}}} --(7/5 * \w,0.5) node [above] {$7/5+\varepsilon$};
\draw(3/2 * \w,-0.5) node [below]{\cite{JansenT10}} --(3/2 * \w,0.5) node [above] {$3/2+\varepsilon$};
\draw[very thick][->](12/11 * \w,0.25) -- node[above]{improvement}(5/4 * \w,0.25);
\draw[pattern = north west lines] (5/4 * \w,-0.25) rectangle (4/3 * \w,0.25);
\node at (31/24 *\w,0.5){open};
\end{tikzpicture}
\caption{The upper and lower bounds for the best possible approximation for pseudo-polynomial Strip Packing achieved so far}
\label{fig:improvement}
\end{figure}
\paragraph*{Notation}

For a given schedule $\sched$ we define for $i \in \jobs$ and any set of jobs $\jobs' \subseteq \jobs$ the value $\#_{i} \jobs'$ as the number of jobs in $\jobs'$, which finish before $\sched(i)$ (e.i.  $\#_{i} \jobs' = |\{j \in \jobs': \sched(j) + p(j) \leq \sched(i)\}|$). 
If the job is clear from the context we write $\# \jobs'$ instead of $\#_{i} \jobs'$. 
Furthermore, we will use a notation defined in \cite{DuL89a} for swapping a part of the content of two machines.
Let $i \in \jobs$ be a job, that is processed by at least two machines $\tilde{M}$ and $\tilde{M}'$ with start point $\sched(i)$. We can swap the content of the machines $\tilde{M}$ and $\tilde{M}'$ after time $\sched(i)$ without violating any scheduling constraint. We define this swapping operation as $SWAP(\sched(i),\tilde{M},\tilde{M}')$.

\paragraph*{Organization of this Paper}

In Section \ref{sec:HardnessScheduling} we will prove that $P4|size_j|C_{\max}$ is strongly NP-complete by a reduction from the strongly NP-complete Problem 3-Partition. 
First, we describe the jobs to construct for this reduction. Afterward, we prove: if the 3-Partition instance is a Yes-instance, then there is a schedule with a specific makespan, and if there is a schedule with this specific makespan then the 3-Partition instance has to be a Yes-instance. While the first can be seen directly, the proof of the second is more involved. 
Proving the second claim, we first show that it can be w.l.o.g. supposed that each machine contains a certain set of jobs.
In the next step, we prove some implications on the order in which the jobs appear on the machines which finally leads to the conclusion that the 3-Partition instance has to be a Yes-instance.
In Section \ref{sec:HardnessStripPacking} we discuss the implications for the inapproximability of pseudo-polynomial Strip Packing.

\section{Hardness of Scheduling Parallel Tasks}
\label{sec:HardnessScheduling}
In the following, we will prove Theorem \ref{thm:SchedulingHardness} by a reduction from the 3-Partition problem. 
%We will show that if we can decide for a number $W \in \mathbb{N}$ and an instance of $P4|size_j|C_{\max}$ whether there exists a schedule with makespan at most $W$ in pseudo-polynomial time, then we can decide the 3-Partition problem in this time as well.
In the 3-Partition problem we have given a list $\mathcal{I} = (\iota_1,\dots,\iota_{3z})$ of $3z$ positive integers, such that $\sum_{i = 1}^{3z}\iota_i = zD$ and $D/4 < \iota_i < D/2$ for each $1\leq i \leq 3z$. The problem is to decide whether there exists a partition of the set $I = \{1,\dots,3z\}$ into sets $I_1, \dots I_z$, such that $\sum_{i \in I_j} \iota_i = D$ for each $1 \leq j \leq z$. We define $SIZE(\mathcal{I}) = \sum_{i = 1}^{3z} \log(\iota_i)$ as the input size of the problem.
3-Partition is strongly NP-complete \cite{GareyJ79}. Therefore, it can not be solved in pseudo-polynomial time, unless $P = NP$. 

\paragraph*{Construction}
First, we will describe how we generate an instance of $P4|size_j|C_{\max}$ from a given 3-Partition instance $\mathcal{I}$ in polynomial time. Let $\mathcal{I}  = (\iota_1,\dots,\iota_{3z})$ be a 3-Partition instance with $\sum_{i = 1}^{3z}\iota_i = zD$. If $D \leq \scalePart$, we scale each number with $\scalePart$ such that we get a new instance $\mathcal{I} ' := (\scalePart \cdot \iota_1,\dots,\scalePart \cdot \iota_{3z})$. For this instance, it holds that $D' = \scalePart D> \scalePart$ and $SIZE(\mathcal{I} ') \in \mathrm{poly}(SIZE(\mathcal{I} ))$. Furthermore, $\mathcal{I}$ is a Yes-instance if and only if $\mathcal{I}'$ is a Yes-instance. Therefore, we can w.l.o.g. assume that $D > \scalePart$.

\begin{figure}
\begin{align*}
p(i) = 
\begin{cases}
D^2 & i \in A\\
D^3 & i \in B\\
D^4 + D^6 +3zD^7 & i \in a,\\
D^5 + D^6 +3zD^7 & i \in b,\\
(z+j)D^7 + D^8 & i = c_j\in c, j\in \{0,\dots,z\}\\
D^3 + D^5 + 4zD^7 +D^8 &  i \in \alpha\\
D^2 + D^4 + (4z-1)D^7 +D^8 &  i \in \beta\\
D^5 + (3z -j)D^7 -D & i = \gamma_j \in \gamma, j\in \{1,\dots,z\}\\
D^4 + (3z -j)D^7 & i = \delta_j \in \delta, j\in \{1,\dots,z\}\\
D^3 + zD^7 + D^8 &  i = \lambda_1\\
D^2 + 2zD^7 + D^8 &  i = \lambda_2\\
\end{cases}
\end{align*}
\caption{Overview of the structure jobs}
\label{fig:overviewItems}
\end{figure}

In the following, we describe the jobs constructed for the reduction; see Figure \ref{fig:overviewItems} for an overview. We generate two sets $A$ and $B$ of $3$-processor jobs. $A$ contains $z+1$ jobs with processing time $p_A := D^2$ and $B$ contains $z+1$ jobs with processing time $p_B := D^3$. 
Furthermore, we generate three sets $a$, $b$ and $c$ of $2$-processor jobs, such that $a$ contains $z$ jobs with processing time $p_a := D^4 + D^6 +3zD^7$, $b$ contains $z$ jobs with processing time $p_b := D^5 + D^6 + 3zD^7$ while $c$ contains one job $c_j$ for each $0 \leq j \leq z$, having processing time $(z+j)D^7 + D^8$ resulting in $z+1$ jobs total in $c$. 
Last we define five sets $\alpha$, $\beta$, $\gamma$, $\delta$, and $\lambda$ of $1$-processor jobs, such that
$\alpha$ contains $z$ jobs with processing time $p_{\alpha} := D^3 + D^5 + 4zD^7 +D^8$, 
$\beta$ contains $z$ jobs with processing time $p_{\beta} := D^2 + D^4 + (4z-1)D^7 +D^8$, 
$\gamma$ contains for each $1 \leq j \leq z$ one job $\gamma_j$ with processing time $D^5 + (3z - j)D^7 -D$ resulting in $|\gamma| = z$, 
$\delta$ contains for each $1 \leq j \leq z$ one job $\delta_j$ with processing time $ D^4 + (3z -i)D^7$ resulting in $|\delta| = z$, 
and $\lambda$ contains two jobs $\lambda_1$ and $\lambda_2$ with processing times  $p(\lambda_1) := D^3 + zD^7 + D^8$ and $p(\lambda_2) := B+ c_0 = D^2 + 2zD^7 + D^8$.
We call these jobs \textit{structure jobs}.
Additionally, we generate for each $i \in \{1, \dots, 3z\}$ one 1-processor job, called \textit{partition job}, with processing time $\iota_i$. We name the set of partition jobs $P$.
Last, we define $W := (z+1)(D^2 + D^3 +D^8) + z(D^4 + D^5 +D^6) + z(7z + 1)D^7$. Note that the work of the generated jobs adds up to $4 W$.

If we add the processing times of all generated jobs, the largest coefficient before a $D^i$ is at most $\scalePart$. Since $\scalePart < D$, it can never happen that in the total processing time of a set of jobs the value $D^i$, together with its coefficient, influences the coefficient of $D^{i+1}$. Furthermore, if the processing times of a set of jobs add up to a value where one of the coefficients is larger than the coefficients in $W$, it is not possible that in a schedule with no idle time one of the machines contains this set.  

In the following sections, we will prove that there is a schedule with makespan $W$ if and only if the 3-Partition instance is a Yes-instance.

\paragraph*{Partition to Schedule} 
Let $\mathcal{I}$ be a Yes-instance with partition $I_1, \dots, I_z $.
One can easily verify that the \textit{structure jobs} can be scheduled as shown in Figure \ref{fig:packingStructure}. After each job $\gamma_j$, for each $1 \leq j \leq z$, we have a gap with processing time $D$. We schedule the \textit{partition jobs} with indices out of $I_j$ directly after $\gamma_j$. Their processing times add up to $D$, and therefore they fit into the gap. The resulting schedule has a makespan of $W$.

\paragraph*{Schedule to Partition}
In this section, we will show that if there is a schedule with makespan $W$, then $\mathcal{I}$ is a Yes-instance.
Let a schedule $S$ with makespan $W$ be given. We will now step by step describe why  $\mathcal{I}$ has to be a Yes-instance. In the first step, we will show that we can transform the schedule, such that certain machines contain certain jobs.

\begin{figure}
\centering
\begin{tikzpicture}
\pgfmathsetmacro{\h}{0.5}
\pgfmathsetmacro{\w}{0.5}

\draw (-3*\w,0*\h) rectangle node[midway]{$\lambda_1$}(0*\w,1*\h);

\foreach \x/\y in {0/1,1/2}{ 
\begin{scope}[xshift = 9*\h*\x cm]
\draw (-2*\w,3*\h) rectangle node[midway]{$\beta_{\y}$}(3*\w,4*\h);
\draw (0,0) rectangle node[midway]{$A_{\x}$}(1*\w,3*\h);
\draw (1*\w,0) rectangle node[midway]{$a_{\y}$}(4*\w,2*\h);
\draw (3*\w,2*\h) rectangle node[midway]{$b_{\y}$}(6*\w,4*\h);
\draw (-3*\w,1*\h) rectangle node[midway]{$B_{\x}$}(-2*\w,4*\h);
\draw (1*\w,2*\h) rectangle node[midway]{$\delta_{\y}$}(3*\w,3*\h);
\draw (4*\w,1*\h) rectangle node[midway]{$\gamma_{\y}$}(5.5*\w,2*\h);
\draw[pattern = north west lines] (5.5*\w,1*\h) rectangle (6*\w,2*\h);
\draw (4*\w,0*\h) rectangle node[midway]{$\alpha_{\y}$}(9*\w,1*\h);
\draw (-2*\w,1*\h) rectangle node[midway]{$c_{\x}$}(0*\w,3*\h);
\end{scope}
}

\node at (17*\w, 2.5*\h) {$\dots$};

\begin{scope}[xshift = 2*\w cm, yshift = -5*\h cm]
\foreach \z/\y/\i in {z-2/-2/z-1,z-1/-1/z} { 
\pgfmathsetmacro{\x}{\y +2}
\begin{scope}[xshift = 9*\w*\x cm]
\draw (-2*\w,3*\h) rectangle node[midway]{$\beta_{\i}$}(3*\w,4*\h);
\draw (0,0) rectangle node[midway, rotate=90]{$A_{\z}$}(1*\w,3*\h);
\draw (1*\w,0) rectangle node[midway]{$a_{\i}$}(4*\w,2*\h);
\draw (3*\w,2*\h) rectangle node[midway]{$b_{\i}$}(6*\w,4*\h);
\draw (-3*\w,1*\h) rectangle node[midway,rotate=90]{$B_{\z}$}(-2*\w,4*\h);
\draw[pattern = north west lines] (5.5*\w,1*\h) rectangle (6*\w,2*\h);
\draw (4*\w,0*\h) rectangle node[midway]{$\alpha_{\i}$}(9*\w,1*\h);
\draw (-2*\w,1*\h) rectangle node[midway]{$c_{\z}$}(0*\w,3*\h);
\end{scope}
}

\begin{scope}[xshift = 9*\w*2 cm]
\draw (-3*\w,1*\h) rectangle node[midway]{$B_z$}(-2*\w,4*\h);
\draw (-8*\w,2*\h) rectangle node[midway]{$\delta_{z}$}(-6*\w,3*\h);
\draw (-5*\w,1*\h) rectangle node[midway]{$\gamma_{z}$}(-3.5*\w,2*\h);
\draw (-2*\w,1*\h) rectangle node[midway]{$c_{z}$}(0*\w,3*\h);
\draw (0,0) rectangle node[midway]{$A_z$}(1*\w,3*\h);
\end{scope}

\draw (1*\w,2*\h) rectangle node[midway]{$\delta_{z-1}$}(3*\w,3*\h);
\draw (4*\w,1*\h) rectangle node[midway]{$\gamma_{z-1}$}(5.5*\w,2*\h);
\end{scope}

\draw (18*\w,-2*\h) rectangle node[midway]{$\lambda_2$}(21*\w,-1*\h);
\node at (-2*\w, - 2.5*\h) {$\dots$};
%\draw (25*\w,3*\h) rectangle node[midway]{$\lambda_2$}(28*\w,4*\h);
%\draw (27*\w,0) rectangle node[midway]{$A$}(28*\w,3*\h);

\node at (-4.5*\w,0.5 *\h) {$M_1$};
\node at (-4.5*\w,1.5 *\h) {$M_2$};
\node at (-4.5*\w,2.5 *\h) {$M_3$};
\node at (-4.5*\w,3.5 *\h) {$M_4$};
\end{tikzpicture}
\caption{An optimal schedule, for a Yes-instance for $a_i \in a$, $b_i \in b$, $A_j \in A$, $B_j \in B$, $\alpha_i \in \alpha$ and $\beta_i \in \beta$.}
\label{fig:packingStructure}
\end{figure}
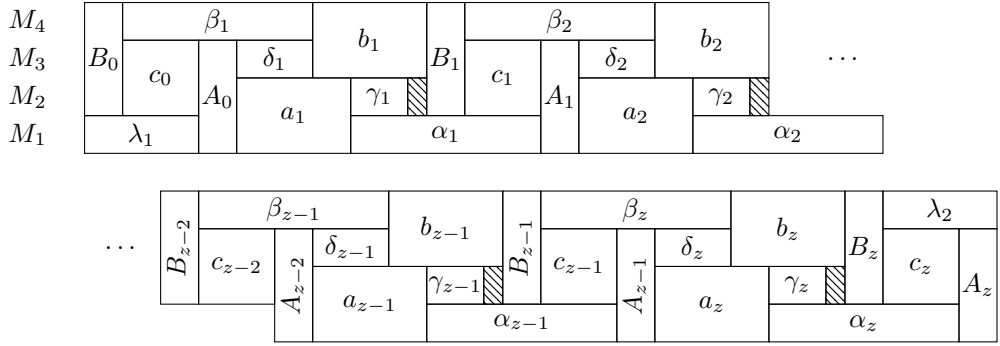

\begin{lemma}
\label{lma:JobsOnMachines}
We can transform the schedule $S$ into a schedule, where
$M_1$ contains the jobs $A \cup a \cup \alpha \cup \lambda_1$, $M_2$ contains the jobs $A \cup B \cup c \cup \check{a} \cup \check{b} \cup \check{\gamma} \cup \check{\delta}$, $M_3$ contains the jobs $A \cup B \cup c \cup \hat{a} \cup \hat{b} \cup \hat{\gamma} \cup \hat{\delta}$ and $M_4$ contains the jobs $B \cup b \cup \beta \cup \lambda_2$, with $\check{a} \subseteq a$, $\hat{a} = a \setminus \check{a}$, $\check{b} \subseteq b$, $\hat{b} = b \setminus \check{b}$, $\check{\gamma} \subseteq \gamma$, $\hat{\gamma} = \gamma \setminus \check{\gamma}$, and $\check{\delta} \subseteq \delta$, $\hat{\delta} = \delta \setminus \check{\delta}$. Furthermore, if the jobs are scheduled in this way, it holds that $|\check{a}| = |\check{\gamma}|$ and $|\check{b}| = |\check{\delta}|$.  
\end{lemma}

\begin{proof}
First, we will show that the content of the machines can be swapped without enlarging the makespan, such that $M_2$ and $M_3$ each contain all the jobs in $A \cup B$. 
Let $x \in A \cup B$ be the job with the smallest starting point in this set. 
We can swap the complete content of the machines such that $M_2$ and $M_3$ contain $x$. 
Let us suppose that, after some swapping operations, $M_2$ and $M_3$ contain the first $i$ jobs in $A \cup B$. 
Let $\tilde{M} \in \{M_1,M_4\}$ be the third machine containing the $i$-th job $x_i \in A \cup B$. 
Let $\tilde{M}'$ be the machine not containing the $(i+1)$-th job. If $\tilde{M}' \in \{M_2,M_3\}$, we transform the schedule such that $M_2$ and $M_3$ contain it, by performing one more swapping operation $SWAP(\sched(x_i),\tilde{M}, \tilde{M}')$. Therefore, we can transform the given schedule without increasing its makespan such that $M_2$ and $M_3$ each contain all the jobs in $A \cup B$. 

In the next step, we will determine the set of jobs contained by the machines $M_1$ and $M_4$. 
The machines $M_2$ and $M_3$ contain besides the jobs in $A \cup B$ jobs with total processing time of $zD^4 + zD^5 +zD^6 + z(7z+1)D^7 + (z+1)D^8$. 
Hence, $M_2$ or $M_3$ can not contain jobs in $\alpha \cup \beta \cup \lambda$, since
their processing times contain $D^2$ or $D^3$. 
Therefore, each job in $A \cup B \cup \alpha \cup \beta \cup \lambda$ is either processed on $M_1$ or on $M_4$. 
In addition to these jobs, $M_1$ and $M_4$ contain further jobs with a total processing time of $zD^4 + zD^5 + 2zD^6 + 6z^2D^7$ total.
The only jobs with a processing time containing $D^6$ are the jobs in the set $a \cup b$. 
Therefore, each machine processes $z$ jobs from the set $a \cup b$.
Hence, a total processing time of $3z^2D^7$ is used by jobs in the set $a \cup b$ on each machine. 
This leaves a processing time of $(4z^2 +z)D^7$ for the jobs in $\alpha \cup \beta \cup \lambda$ on $M_1$ and $M_4$ corresponding to $D^7$. 
All the $2(z+1)$ jobs in $\alpha \cup \beta \cup \lambda$ contain $D^8$ in their processing time. 
Therefore, each machine $M_1$ and $M_4$ processes $z+1$ of them. 
We will swap the content of $M_1$ and $M_4$ such that $\lambda_1$ is scheduled on $M_1$. 
As a consequence, $M_1$ processes $z$ jobs from the set $\alpha \cup \beta \cup \{\lambda_2\}$, with processing times, which sum up to $4z^2 D^7$ in the $D^7$ component.
The jobs in $\alpha$ have with $4zD^7$ the largest amount of $D^7$ in their processing time. 
Therefore, $M_1$ processes all of them since $z \cdot 4zD^7 = 4z^2D^7$, while $M_4$ contains the jobs in $\beta \cup \{\lambda_2\}$. 
Since we have $p(\alpha \cup \{\lambda_1\}) = (z+1)D^3 + zD^5 +z(4z+1)D^7(z+1)D^8$, jobs from the set $A \cup B \cup a \cup b$ with total processing time of $(z+1)D^2 + zD^4 + zD^6 + 3z^2D^7$ have to be scheduled on $M_1$. In this set, the jobs in $A$ are the only jobs with processing times containing $D^2$, while the jobs in $a$ are the only jobs with a processing time containing $D^4$.  
As a consequence, $M_1$ processes the jobs $A \cup a \cup \alpha \cup \{\lambda_1\}$. Analogously we can deduce that $M_4$ processes the jobs $B \cup b \cup \beta \cup \{\lambda_2\}$. 

In the last step, we will determine which jobs are scheduled on $M_2$ and $M_3$. 
As shown before, each of them contains the jobs $A \cup B$. 
Furthermore, since no job in $c$ is scheduled on $M_1$ or $M_4$, and they require two machines to be processed, machines $M_2$ and $M_3$ both contain the set $c$. 
Additionally, each job in $\gamma \cup \delta$ has to be scheduled on $M_2$ or $M_3$ since they are not scheduled on $M_1$ or $M_4$. Each job in $a \cup b$ occupies one of the machines $M_1$ and $M_4$. The second machine they occupy is either $M_2$ or $M_3$. Let $\check{a} \subseteq a$ be the set of jobs, which is scheduled on $M_2$ and $\hat{a} \subseteq a$ be the set which is scheduled on $M_3$. Clearly $\check{a} = a \setminus \hat{a}$. We define the sets $\hat{b}, \check{b}, \hat{\delta},\check{\delta}, \hat{\gamma}$, and $\check{\gamma}$ analogously. By this definition, $M_2$ contains the jobs $A \cup B \cup \check{a} \cup \check{b} \cup \check{\delta} \cup \check{\gamma} \cup c$ and $M_3$ contains the jobs $A \cup B \cup \hat{a} \cup \hat{b} \cup \hat{\delta} \cup \hat{\gamma} \cup c$. 

We still have to show that $|\check{a}| = |\check{\gamma}|$ and $|\check{b}| = |\check{\delta}|$. 
First, we notice that $|\check{a}| + |\check{b}| = z$ since these jobs are the only jobs with a processing time containing $D^6$. 
So besides the jobs in $A \cup B\cup c \cup \check{a} \cup \check{b}$, $M_2$ contains jobs with total processing time of $(z-|\check{a}|)D^4 + (z-|\check{b}|)D^5 + \sum_{i = 1}^z(3z-i)D^7 =|\check{b}| D^4 + |\check{a}|D^5 + \sum_{i = 1}^z(3z-i)D^7$. Since the jobs in $\delta$ are the only jobs in $\delta \cup \gamma$ having a processing time containing $D^4$, we have $|\check{\delta}|= |\check{b}|$ and analogously $|\check{\gamma}|= |\check{a}|$.
\end{proof}

In the next steps, we will prove that it is possible to transform the order in which the jobs appear on the machines to the order in Figure \ref{fig:packingStructure}. 
Notice that, since there is no idle time in the schedule, each start point of a job $i$ is given by the sum of processing times of the jobs on the same machine scheduled before $i$. So the start position $\sched(i)$ of a job $i$ has the form 
\[\sched(i) = x_0 + x_2D^2 + x_3D^3 +x_4D^4 + x_5D^5 + x_6D^6 + x_7D^7 + x_8D^9\]
for $-zD \leq x_0 \leq zD$ and $0 \leq x_j \leq 4z(7z + 1) \leq D$ for each $2 \leq j \leq 8$. This allows us to make implications about the correlation between the number of jobs scheduled on different machines when a job from the set $A \cup B\cup a \cup b \cup c$ starts. 
For example, let us look at the coefficient $x_2$. This value is just influenced by jobs with processing times containing $D^2$. 
The only jobs with these processing times are the jobs in the set $A \cup \beta \cup \{\lambda_2\}$. The jobs in $\beta \cup \{\lambda_2\}$ are just processed on $M_4$, while the jobs in $A$ each are processed on the three machines $M_1$, $M_2$, and $M_3$. 
Therefore, we know that at the starting point $\sched(i)$ of a job $i$ scheduled on machines $M_1$, $M_2$ or $M_3$ we have that $x_2 = \#_iA$. 
Furthermore, if $i$ is scheduled on $M_4$ we know that $x_2 = \#_i\beta + \#_i\{\lambda_2\}$.
In Table \ref{table:coefficients} we present which sets influences which coefficients in which way when job $i$ is started on the corresponding machine.

\begin{table}
\center
\begin{tabular}{r|cccc}
&$M_1$&$M_2$&$M_3$ & $M_4$\\
\hline
$x_2$&$\#_iA$&$\#_iA$&$\#A$&$\#_i\beta + \#_i\{\lambda_2\}$\\
$x_3$&$\#_i\alpha + \#_i\{\lambda_1\}$&$\#_iB$&$\#_iB$ & $\#_iB$ \\
$x_4$&$\#_i a$&$\#_i \check{a} + \#_i\check{\delta}$&$\#_i \hat{a} + \#_i\hat{\delta}$& $\#_i\beta $\\
$x_5$&$\#_i \alpha$&$\#_i \check{b} + \#_i\check{\gamma}$&$\#_i \hat{b} + \#_i\hat{\gamma}$& $\#_i b$\\
$x_6$&$\#_i a$&$\#_i \check{a} + \#_i\check{b}$&$\#_i \hat{a} + \#_i\hat{b}$& $\#_i b$\\
$x_8$&$\#_i\alpha + \#_i\{\lambda_1\}$&$\#_ic$&$\#_ic$& $\#_i\beta + \#_i\{\lambda_2\}$\\
\end{tabular}
\caption{Overview of the values of the coefficients at the start point of a job $i$, if $i$ is scheduled on machine $M_j$.}
\label{table:coefficients}
\end{table}
Let us consider the start point $\sched(i)$ of a job $i$, which uses more than one machine.
We know that $\sched(i)$ is the same on all the used machines and therefore the coefficients are the same as well.  
In the following, we will study for each of the sets $A$, $B$, $a$, $b$, $c$ what we can conclude for the starting times of these jobs. For each of the sets, we will present an equation, which holds at the start of each item in this set. These equations give us a strong set of tools for our further arguing.

First, we will consider the start points of the jobs in $A$. 
Each job $A' \in A$ is scheduled on machines $M_1$, $M_2$ and $M_3$. 
Therefore, we know that at $s(A')$ we have $\#_{A'}B =_{x_3} \#_{A'}\alpha + \#_{A'}\{\lambda_1\} =_{x_8} \#_{A'} c$. 
Furthermore, we know that $\#_{A'} a =_{x_6} \#_{A'} \check{a} + \#_{A'} \check{b} = \#_{A'} \hat{a} + \#_{A'} \hat{b}$. 
Since $\#_{A'} a = \#_{A'} \check{a} + \#_{A'} \hat{a}$ and $\#_{A'} b = \#_{A'} \check{b} + \#_{A'} \hat{b}$, we can deduce that $\#_{A'} a = \#_{A'} b$. Additionally, we know that $\#_{A'} \alpha  =_{x_5} \# \check{b} + \#\check{\gamma} =_{x_5} \# \hat{b} + \#\hat{\gamma}$. 
Thanks to this equality, we can show that $\#_{A'} \alpha = \#_{A'} b$. 
First, we show $\#_{A'} \alpha \geq \#_{A'} b$.
Let $b' \in b$ be the last job in $b$ scheduled before $A'$ if there is any. 
Let us w.l.o.g assume that $b \in \hat{b}$. 
It holds that $\#_{A'} b = \#_{b'} b +1 =_{x_5} \#_{b'} \hat{b} + \#_{b'} \hat{\gamma} + 1 \leq \#_{A'} \hat{b} + \#_{A'} \hat{\gamma} =_{x_5} \#_{A'} \alpha$. 
If there is no such $b'$ we have $\#_{A'} b = 0 \leq \#_{A'} \alpha$.
Next, we show $\#_{A'} \alpha \leq \#_{A'} b$.
Let $b'' \in A$ be the first job in $b$ scheduled after $A$ if there is any. 
Let us w.l.o.g assume that $b \in \check{b}$. 
It holds that $\#_{A'} b = \#_{b''} b =_{x_5} \#_{b''} \check{b} + \#_{b''} \check{\gamma}\geq \#_{A'} \check{b} + \#_{A'} \check{\gamma} =_{x_5} \#_{A'} \alpha$. 
If there is no such $b''$, we have $\#_{A'} b = z \geq \#_{A'} \alpha$. 
As a consequence we have $\#_{A'} \alpha = \#_{A'} b$. 
In summary, we can deduce that 
\begin{equation}
\#_{A'} c - \#_{A'}\{\lambda_1\} = \#_{A'}B - \#_{A'}\{\lambda_1\} =\#_{A'} \alpha = \#_{A'} b = \#_{A'} a.
\label{eq:A1}
\end{equation}
Analogously, we can deduce that at the start of each $B' \in B$ we have that
\begin{equation}
\#_{B'} c - \#_{B'}\{\lambda_2\} = \#_{B'}A - \#_{B'}\{\lambda_2\} = \#_{B'} \beta = \#_{B'} a = \#_{B'} b.
\label{eq:B1}
\end{equation}
Each item $a' \in a$ is scheduled on machine $M_1$ and on one of the machines $M_2$ or $M_3$. For each possibility, we can deduce the equation
\begin{equation}
\label{eq:a1}
\#_{a'}B =_{x_3} \#_{a'}\alpha + \#_{a'}\{\lambda_1\} =_{x_8} \#_{a'} c.
\end{equation}
Analogously, we deduce for each $b' \in b$ that
\begin{equation}
\label{eq:b1}
\#_{b'}A =_{x_2} \#_{b'}\beta + \#_{b'}\{\lambda_2\} =_{x_8} \#_{b'} c.
\end{equation} 
Last, each item $c' \in c$ is scheduled on $M_2$ and $M_3$.
Let $a' \in a$ be the job with the smallest $\sched(a') \geq \sched(c')$. Let us w.l.o.g assume that $a' \in \hat{a}$.
It holds that 
$\#_{c'} \check{a} + \#_{c'}\check{b} =_{x_6} 
\#_{c'} \hat{a} + \#_{c'}\hat{b} \leq 
\#_{a'} \hat{a} + \#_{a'}\hat{b} = _{x_6}
\#_{a'} a = \#_{a'} \hat{a} + \#_{a'}\check{a} = \#_{c'} \hat{a} + \#_{c'}\check{a}$. As a consequence, we have $\#_{c'}\check{b} \leq \#_{c'} \hat{a}$ and  $\#_{c'}\hat{b} \leq \#_{c'} \check{a}$. Analogously, let $b ' \in b$ be the job with the smallest $\sched(b') \geq \sched(c')$. Let us w.l.o.g assume that $b' \in \check{b}$.
It holds that 
$\#_{c'} \hat{a} + \#_{c'}\hat{b} =_{x_6} 
\#_{c'} \check{a} + \#_{c'}\check{b} \leq 
\#_{b'} \check{a} + \#_{b'}\check{b} = _{x_6}
\#_{b'} b = \#_{b'} \hat{b} + \#_{b'}\check{a} = \#_{c'} \hat{b} + \#_{c'}\check{b}$. 
Therefore, $\#_{c'}\check{a} \leq \#_{c'} \hat{b}$ and  $\#_{c'}\hat{a} \leq \#_{c'} \check{b}$.
As a consequence, we can deduce that
\begin{equation}
\label{eq:c1}
\#_{c'} b = \#_{c'} a
\end{equation}

These equations give us the tools to analyze the given schedule with makespan $W$. First, we will show that in this schedule the first and last jobs have to be elements from the set $A \cup B$, (see Lemma \ref{lma:firstJob}). After that, we will prove that the jobs in $A$ and jobs in $B$ have to be scheduled alternating, (see Lemma \ref{lma:AeqB}). 
With the knowledge gathered in the proofs of Lemma \ref{lma:firstJob} and Lemma \ref{lma:AeqB}, we can prove that the given schedule can be transformed such that all jobs are scheduled continuously, and that $\mathcal{I}$ has to be a Yes-instance (see Lemma \ref{lma:AeqB}).

\begin{lemma}
\label{lma:firstJob}
The first and the last job on $M_2$ and $M_3$ are elements of $A \cup B$.
\end{lemma}
\begin{proof}
Let $i := \arg\min_{i \in A \cup B} s_i$ be the job with the smallest start point in $A \cup B$, (i.e. $\#_{i} A = 0 = \#_{i} B$). We have to consider each case $i \in A$ and $i \in B$ and to show that its starting time has the value $s_i = 0$. 

If $i \in A$ it holds that $0= \#_{i}B =_{(\ref{eq:A1})} \#_{i}\alpha + \#_{i}\{\lambda_1\} =_{(\ref{eq:A1})} \#_{i} a + \#_{i}\{\lambda_1\}$ and therefore $\#_{i} a =\#_{i}\alpha = 0 = \#_{i}\{\lambda_1\}$. 
The jobs $a \cup \alpha \cup \{\lambda_1\} \cup A$ are the only jobs, which are contained on machine $M_1$. 
Since $\#_{i} A =0$ as well, it has to be that $s_i = 0$ and therefore $i$ is the first job on $M_2$ and $M_3$.

If $i \in B$ it holds that $ 0= \#_{i}A =_{(\ref{eq:B1})}\#_{i}\beta + \#_{i}\{\lambda_2\} =_{(\ref{eq:B1})} \#_{i} b+ \#_{i}\{\lambda_2\}$ and therefore $\#_{i} b = \#_{i}\beta = 0 =\#_{i}\{\lambda_2\}$. The jobs $b \cup \beta \cup \{\lambda_2\} \cup B$ are the only jobs, which are contained on machine $M_4$. Since $\#_{i} B = 0$ as well, it has to be that $s_i = 0$ and therefore $i$ is the first job on $M_2$ and $M_3$.

We have shown that the first job on $M_2$ and $M_3$ hast to be a job from the set $A \cup B$. Since the schedule stays valid, if we mirror the schedule such that the new start points are $s'(i) = W-s(i)-p(i)$ for each job $i$, the last job has to be in the set $A \cup B$ as well.
\end{proof}

Next, we will show that the items in the sets $A$ and $B$ have to be scheduled alternating. 
Let $(A_0, \dots, A_z)$ be the set $A$ and $(B_0, \dots, B_z)$ be the set $B$ each ordered by increasing size of the starting points.

\begin{lemma}
\label{lma:AeqB}
If the first item on $M_2$ is the job $B_0 \in B$ it holds for each item $i \in \{0, \dots, z\}$ that
\begin{equation}
\label{eq:AandB}
\#_{A_i} B - \#_{A_i} \{\lambda_1\} = \#_{A_i} A
\end{equation}
with $\#_{A_i} \{\lambda_1\} =1$. 
\end{lemma}

\begin{proof}
We will prove this claim inductively and per contradiction. 

Assume $\#_{A_0} B - \#_{A_0} \{\lambda_1\} > \#_{A_0} A = 0$. 
Therefore, we have 
$1 \leq \#_{A_0} B - \#_{A_0} \{\lambda_1\}.$ 
Let $a' \in a$, $b' \in b$ and $c' \in c$ be the first started jobs in their sets.
Since $\#_{A_0} b =_{(\ref{eq:A1})} \#_{A_0} a =_{(\ref{eq:A1})} \#_{A_0} c - \#_{A_0} \{\lambda_1\}=_{(\ref{eq:A1})} \#_{A_0} B - \#_{A_0} \{\lambda_1\} \geq 1$, the jobs $a',b'$ and $c'$ start before $A_0$.
It holds that
$\#_{b'} c =_{(\ref{eq:b1})} \#_{b'} A = 0.$
Therefore, $c'$ has to start after $b'$ resulting in $\#_{c'}b  \geq 1$. 
Furthermore, we have 
$\#_{a'} c =_{(\ref{eq:a1})} \#_{a'} B \geq 1.$
Hence, $c'$ has to start before $a'$ resulting in $\#_{c'}a = 0$. 
In total we have
$1 \leq \#_{c'}b =_{(\ref{eq:c1})} \#_{c'}a = 0$
contradicting the assumption $\#_{A_0} B - \#_{A_0} \{\lambda_1\} > \#_{A_0} A = 0$.
Therefore, we have $\#_{A_0} B - \#_{A_0} \{\lambda_1\} \leq \#_{A_0} A = 0$. As a consequence, it holds that $1 \leq \#_{A_0} B \leq \#_{A_0} \{\lambda_1\} \leq 1$ and we can conclude $\#_{A_0} B = 1 = \#_{A_0} \{\lambda_1\}$ as well as $\#_{A_0} B - \#_{A_0} \{\lambda_1\} = \#_{A_0} A$.

Choose $i \in \{0,\dots, z\}$ such that $\#_{A_{i'}} B - \#_{A_{i'}} \{\lambda_1\} = \#_{A_{i'}} A$ for all $i' \in \{0, \dots i\}$.
As a consequence, we have $\#_{B_{i}} B = i = \#_{A_{i}} A = \#_{A_{i}} B -1$. Therefore $B_i$ has to be scheduled before $A_i$. Additionally, we have $\#_{B_{i}} B-1=\#_{B_{i-1}} B= i-1=\#_{A_{i-1}} A = \#_{A_{i-1}} B -1$, so $B_i$ has to be scheduled after $A_{i-1}$. Therefore, we have $\#_{B_{i}} B = \#_{B_{i}} A$ and as a consequence 
\begin{equation}
i = \#_{B_{i}} B = \#_{B_{i}} A = \#_{B'} c = \#_{B'} \beta + \#_{B'}\{\lambda_2\}= \#_{B'} a + \#_{B'}\{\lambda_2\} = \#_{B'} b + \#_{B'}\{\lambda_2\}.
\label{eq:AB}
\end{equation}
We will now prove our claim for $A_{i+1}$.

\begin{claim}
$\#_{A_{i+1}} B - \#_{A_{i+1}} \{\lambda_1\} \leq \#_{A_{i+1}} A$
\end{claim}

Assume for contradiction that $\#_{A_{i+1}} B - \#_{A_{i+1}} \{\lambda_1\} > \#_{A_{i+1}} A$. 
As a consequence, we have 
$\#_{A_{i+1}} B - \#_{A_{i+1}} \{\lambda_1\} - \#_{A_{i}} B + \#_{A_{i}} \{\lambda_1\} \geq 2$. Therefore, there are jobs $B_{i+1}, B_{i+2} \in B$,  $a', a'' \in a$,  $b', b'' \in b$ and $c', c'' \in c$, that are scheduled between $A_i$ and $A_{i+1}$ since equality (\ref{eq:A1}) holds. Let us suppose that $\sched(a')\leq \sched(a'')$, $\sched(b')\leq \sched(b'')$ and $\sched(c')\leq \sched(c'')$.

Next, we will deduce in which order the jobs $a',a'',b',b'',c',c'',B_{i+1}$, and $B_{i+2}$ appear in the schedule.
It holds that 
$ \#_{b''} c =_{(\ref{eq:b1})} \#_{b''} A = \#_{A_i} A +1 =_{}\#_{A_i} B  =_{(\ref{eq:A1})} = \#_{A_i} c.$
Therefore, $b'$ and $b''$ have to start before $c'$. Furthermore we have
$\#_{c'} a =_{(\ref{eq:c1})} \#_{c'} b \geq \#_{A_{i}} b +2 =_{(\ref{eq:A1})} \#_{A_{i}} a +2.$
Hence, $a''$ hast to start before $c'$ as well. Additionally, it holds that
$\#_{B_{i+2}} c =_{(\ref{eq:B1})} \#_{B_{i+2}} A = \#_{A_{i}} A +1 =_{} \#_{A_{i}} B =_{(\ref{eq:A1})} \#_{A_{i}} c.$ 
As a consequence, $B_{i+2}$ has to start before $c'$. Additionally, $a''$ has to start before $B_{i+1}$, since
$\#_{a''} B =_{(\ref{eq:a1})} \#_{a''} c = \#_{A_{i}} c =_{(\ref{eq:A1})} \#_{A_{i}} B.$

To this point, we have deduced that the jobs have to appear in the following order in the schedule: $A_i, a',a'',B_{i+1},B_{i+2},c',c'',A_{i+1}$. This schedule is not feasible, since we have
$\#_{A_i} a +2 \leq_S \#_{B_{i+1}} a {\leq}_{(\ref{eq:B1})} \#_{B_{i+1}} A =_S \#_{A_i}A + 1 {=}_{(\ref{eq:A1})} \#_{A_i} a +1,$
a contradiction to the assumption $\#_{A_{i+1}} B - \#_{A_{i+1}} \{\lambda_1\} > \#_{A_{i+1}} A$. Therefore, it holds that $\#_{A_{i+1}} B - \#_{A_{i+1}} \{\lambda_1\} \leq \#_{A_{i+1}} A$

\begin{claim}
$\#_{A_{i+1}} B - \#_{A_{i+1}} \{\lambda_1\} \geq \#_{A_{i+1}} A$
\end{claim}

Assume for contradiction that $\#_{A_{i+1}} B - \#_{A_{i+1}} \{\lambda_1\} < \#_{A_{i+1}} A$. It follows that $\#_{A_{i+1}} B = \#_{A_{i}} B$ since 
$\#_{A_{i}} B - \#_{A_{i}} \{\lambda_1\}\leq \#_{A_{i+1}} B - \#_{A_{i+1}} \{\lambda_1\}\leq \#_{A_{i+1}} A -1 = \#_{A_{i}} A = \#_{A_{i}} B - \#_{A_{i}}\{\lambda_1\}.$
Furthermore, there has to be at least one job $B_{i+1} \in B$ that starts after $A_{i+1}$ since $|A| = |B|$. Therefore, we have $\#_{B_{i+1}} c - \#_{B_{i}} c =\#_{B_{i+1}} A - \#_{B_{i}} A \geq 2$. As a consequence, there are jobs $c', c'' \in c$ which are scheduled between $B_i$ and $B_{i+1}$. Let $c'$ be the first job in $c$ scheduled after $B_{i}$ ans $c''$ be the next. Since we do not know the value of $\#_{B_{i}} \{\lambda_2\}$ or $\#_{B_{i +1}} \{\lambda_2\}$, we can just deduce from equation (\ref{eq:B1}) that $\#_{B_{i+1}} a - \#_{B_{i}} a \geq 1$. Therefore, there has to be a job $a' \in a$ that is scheduled between $B_i$ and $B_{i+1}$. 

We will now look at the order in which the jobs $A_i$, $A_{i+1}$, $c'$, $c''$ and $a'$ have to be scheduled. 
First, we know that  $A_i$ and $A_{i+1}$ have to be scheduled between $c'$ and $c''$, since
$\#_{A_{i}} c =_{(\ref{eq:A1})} \#_{A_{i}} B =_S \#_{B_i} B +1 =_{(\ref{eq:AB})} \#_{B_i} A +1 =_{(\ref{eq:B1})} \#_{B_i} c +1$ 
and
$\#_{A_{i+1}} c =_{(\ref{eq:A1})} \#_{A_{i+1}} B =_S \#_{B_i} B +1 =_{(\ref{eq:AB})} \#_{B_i} A +1 =_{(\ref{eq:B1})} \#_{B_i} c +1.$ 
Furthermore, we know that $a'$ has to be scheduled between $c'$ and $c''$ as well, since
$\#_{a'} c =_{(\ref{eq:a1})} \#_{a'} B =_S \#_{B_i} B +1 =_{(\ref{eq:AB})} \#_{B_i} A +1 =_{(\ref{eq:B1})} \#_{B_i} c +1.$ 
As a consequence, we can deduce that there is a job $b' \in b$ which is scheduled between $c'$ and $c''$, since
$\#_{c''} b =_{(\ref{eq:c1})} \#_{c''} a \geq_S \#_{c'} a +1 =_{(\ref{eq:c1})} \#_{c'} b +1.$ 
We know about this $b'$ that 
$\#_{b'} A =_{(\ref{eq:b1})} \#_{b'} c =_S \#_{B_{i}} c +1 =_{(\ref{eq:B1})} \#_{B_{i}} A +1,$
so $b'$ has to be scheduled between $A_i$ and $A_{i+1}$. 

In summary, the jobs are scheduled as follows: $B_i, c', A_i, b', A_{i+1}, c'', B_{i+1}$. 
However, this schedule is infeasible since 
$\#_{A_{i}} b =_{(\ref{eq:A1})} \#_{A_{i}} B - \#_{A_{i}} \{\lambda_1\} =_{S}  \#_{A_{i+1}} B - \#_{A_{i+1}} \{\lambda_1\} =_{(\ref{eq:A1})} \#_{A_{i+1}} b =_{S}  \#_{A_{i}} b +1.$
This contradicts the assumption $\#_{A_{i+1}} B - \#_{A_{i+1}} \{\lambda_1\} < \#_{A_{i+1}} A$. Altogether, we have shown that $\#_{A_{i+1}} B - \#_{A_{i+1}} \{\lambda_1\} = \#_{A_{i+1}} A$.
\end{proof}

A direct consequence of Lemma \ref{lma:AeqB} is that the last job on $M_2$ is a job in $A$.
Since the equations (\ref{eq:A1}) and (\ref{eq:B1}), as well as (\ref{eq:a1}) and (\ref{eq:b1}), are symmetric, we can deduce an analogue statement if the first job on $M_2$ is in $A$. 
More precisely in this case we can show that $\#_iA - \#_i\{\lambda_2\} = \#_i B$ and $\#_i\{\lambda_2\} = 1$ for each $i \in B$.
This would imply that the last job on $M_2$ is a job in $B$.
Since we can mirror the schedule such that the last job is the first job, we can  suppose that the first job on $M_2$ is a job out of $B$. In this case a further direct consequence of Lemma \ref{lma:AeqB} and equation (\ref{eq:A1}) is the equation
\begin{equation}
\label{eq:A}
i = \#_{A_i} A = \#_{A_i}B - 1 = \#_{A_i} c - 1 =\#_{A_i} \alpha = \#_{A_i} b = \#_{A_i} a
\end{equation}

\begin{lemma}
$\mathcal{I}$ is a Yes-instance and we can transform the schedule such that all jobs are scheduled on continuous machines.
\end{lemma}

\begin{proof} 
First, we will show that $\lambda_2$ is scheduled after the last job in $B$.
Assume there is an $i \in \{0,\dots, z\}$ with $\#_{B_i} \{\lambda_2\} > 0$. Let $i$ be the smallest of these indices. 
We know that 
\[i -1 =_{(\ref{eq:AB})}\#_{B_{i}} A - 1 = \#_{B_{i}} A - \#_{B_{i}} \{\lambda_2\}=_{(\ref{eq:B1})} \#_{B_{i}} a.\]
%and therefore $\#_{B_{i}} b =_{(\ref{eq:B1})} \#_{B_{i}} a = i-1$. 
Since $\#_{A_{i}} b =_{(\ref{eq:A1})} \#_{A_{i}} a =_{(\ref{eq:A})} i = \#_{B_{i}} a + 1 =_{(\ref{eq:B1})} \#_{B_{i}} b + 1$ there has to be an unique $a' \in a$ and an unique $b'\in b$ scheduled between $B_i$ and $A_i$. 
Furthermore, since $\#_{A_i} c =_{(\ref{eq:A})} i+1$ and $\#_{B_i} c =_{(\ref{eq:AB})} i$, there has to be a $c' \in c$ scheduled between $B_i$ and $A_i$ as well. 
At the start of $b'$ it holds that $\#_{b'} c =_{(\ref{eq:b1})} \#_{b'} A = \#_{A_{i-1}}A +1 =_{(\ref{eq:A1})} \#_{A_{i-1}}c $, so $b'$ has to start before $c'$. Additionally, at the start of $a'$ we have $\#_{a'} c =_{(\ref{eq:b1})} \#_{a'} B = \#_{B_{i}}B +1 =_{(\ref{eq:AB})} \#_{B_{i}}c +1$ and therefore $a'$ hast to start after $c'$. 
In total, the jobs appear in the following order: $B_i, b',c',a', A_i$. But this can not be the case, since we have $\#_{B_{i-1}} a =_S \#_{c'} a =_{(\ref{eq:c1})} \#_{c'} b =_S \#_{B_{i-1}} b +1 =\#_{B_{i-1}} a +1$. 
Hence, we have contradicted that assumption.
As a consequence, we have $\#_{B_i} \{\lambda_2\} = 0$ for all $i \in \{0,\dots, z\}$ and therefore
\begin{equation}
\label{eq:B}
\#_{B_i} b =\#_{B_i} a = \#_{B_i} c = \#_{B_i} \beta =  \#_{B_i} A = \#_{B_i} B = i. 
\end{equation} 

In the next step, we will prove that $M_1$ processes the jobs $A \cup a \cup \alpha \cup \{\lambda_1\}$ in the order $\lambda_1, A_0, a_1, \alpha_1, A_1, a_2, \alpha_2, A_2, \dots, a_z, \alpha_z, A_z$, where $a_i\in a$ and $\alpha_i \in \alpha$ for each $i \in \{1,\dots z\}$.
%Now, let us look at the jobs on $M_1$. 
Equation (\ref{eq:A}) and Lemma \ref{lma:AeqB} ensure that the first job on $M_1$ is the job $\lambda_1$ and the second job is $A_0$. For each $i \in \{1,\dots,z\}$ it holds that $\#_{A_{i}} \alpha =_{(\ref{eq:A})} \#_{A_{i-1}} \alpha +1$ and $ \#_{A_{i}} a =_{(\ref{eq:A})}\#_{A_{i-1}} a +1$.
Therefore, there is scheduled exactly one job $a_i \in a$ and one job $\alpha_i \in \alpha$ between the jobs $A_{i-1}$ and $A_{i}$.
It holds that $\#_{A_{i-1}} a + 1 =_{(\ref{eq:A})} i =_{(\ref{eq:B})} \#_{B_{i}} a$. Therefore, $a_i$ has to be scheduled between $A_{i-1}$ and $B_{i}$. 
As a consequence, we have
$\#_{a_i}\alpha + 1 = \#_{a_i}\alpha + \#_{a_i}\{\lambda_1\} =_{(\ref{eq:a1})} \#_{a_i}B = \#_{B_{i}}B =_{(\ref{eq:B})} \#_{B_{i}}a = \#_{a'}a +1.$
Therefore, $a_i$ has to be scheduled before $\alpha_i$ and the jobs appear in machine $M_1$ in the described order. As a result, we know about the start point of $A_i$ that 
\begin{align*}
\sched(A_i) &= p(\lambda_1) + i\cdot p_{a} + i \cdot p_{\alpha} + i \cdot p_{A} \\
&= D^3 + zD^7 + D^8 + i(D^4 + D^6 + 3zD^7) + i(D^3 +D^5 + 4zD^7 + D^8) + iD^2\\
& = iD^2 + (i+1)D^3 + iD^4 + iD^5 + iD^6 + (7zi+z)D^7+ (i+1)D^8.
\end{align*}

Now, we will show, that the machine $M_4$ processes the jobs $B \cup b \cup \beta \cup \{\lambda_2\}$ in the order $B_0,\beta_1,b_1,B_1, \beta_2, b_2, B_2, \dots, \beta_z, b_z, B_z, \lambda_2$, where $b_i \in b$ and $\beta_i \in \beta$ for each $i \in \{1,\dots z\}$.
%Let us consider the jobs on $M_4$. 
The first job on $M_4$ is the job $B_0$. 
Equation (\ref{eq:B}) ensures that between the jobs $B_i$ and $B_{i+1}$ there is scheduled exactly one job $b_{i+1} \in b$ and exactly one job $\beta_{i+1} \in \beta$. It holds that $\#_{A_{i}} b + 1 =_{(\ref{eq:A})} i+1 =_{(\ref{eq:B})} \#_{B_{i+1}} b$. Therefore, $b_{i+1}$ has to be scheduled between $A_{i}$ and $B_{i+1}$.
As a consequence, it holds that
$\#_{b_{i+1}}\beta = \#_{b_{i+1}}\beta + \#_{b_{i+1}}\{\lambda_2\} = \#_{b_{i+1}}A = \#_{B_{i+1}}A = \#_{B_{i+1}}b = \#_{b_{i+1}}b +1.$
Hence, $b_{i+1}$ has to be scheduled after $\beta_{i+1}$ and the jobs on machine $M_4$ appear in the described order. As a result, we know about the start point of $B_i$ that 
\begin{align*}
\sched(B_i) &= i p_{b} + i p_{\beta} + ip_B\\
& = iD^2 + iD^3 + iD^4 + iD^5 + iD^6 + (i(7z-1))D^7+ iD^8.
\end{align*}
%&=  i(D^5 + D^6 + 3zD^7) + i(D^2 +D^4 + (4z-1)D^7 + D^8) + (i+1)D^3\\

Next, we can deduce, that the jobs in $c$ are scheduled as shown in Figure \ref{fig:packingStructure}. 
We have $\#_{B_i} c =_{(\ref{eq:B})} i  =_{(\ref{eq:A})} \#_{A_i} c -1$.
Therefore, there exists an $c' \in c$ for each $i \in \{0,\dots, z\}$, which is scheduled between $B_i$ and $A_i$. The processing time between $B_i$ and $A_i$ is exactly $\sched(A_i) - \sched(B_i) - p(B_i)= (z+i)D^7 + D^8$.
As a consequence, one can see with an inductive argument that $c_i \in c$ with $p(c_i) = (z+i)D^7 + D^8$ has to be positioned between $B_i$ and $A_i$, since the job in $c$ with the largest processing time $c_z$ only fits between $B_z$ and $A_z$.

In this step, we will transform the schedule, such that all jobs are scheduled on continuous machines. 
To this point, this property is obviously fulfilled by the jobs in $A \cup B \cup c$. 
However, the jobs in $a \cup b$ might be scheduled on nonconsecutive machines. 
We know that the $a_i$ and $b_i$ are scheduled between $A_{i-1}$ and $B_{i}$. 
One part of $a_i$ is scheduled on $M_1$ and one part of $b_i$ is scheduled on $M_4$, while each second part is scheduled either on $M_2$ or on $M_3$ but both parts on different machines, because $\sched(B_i) - \sched(A_{i-1}) - p(A_i) = D^4 +D^5 + D^6 + (6z-i)D^7 < D^4 + D^5 + 2D^6 + 6zD^7 = p(a_i) +p(b_i)$ for each $i \in \{0,\dots,z\}$.
Since $A_i$ and $B_{i+1}$ both are scheduled on machines $M_2$ and $M_3$, we can swap the content of the machines between these jobs such that the second part of $a_i$ is scheduled on $M_2$ and the second part of $b_i$ is scheduled on $M_3$.
We do this swapping step for all $i \in \{0,\dots,z-1\}$ such that all second parts of jobs in $a$ are scheduled on $M_2$ and all second part of jobs in $b$ are scheduled on $M_3$ respectively. After this swapping step, all jobs are scheduled on continuous machines. 

Now, we will show that $\mathcal{I}$ is a yes-instance.
To this point we know that $M_2$ contains the jobs $A\cup B \cup a \cup c$. 
Since $\check{a} = a$, it has to hold by Lemma \ref{lma:JobsOnMachines}, using $|\check{a}| = |\check{\gamma}|$, that $\check{\gamma} = \gamma$ implying that $M_2$ contains all jobs in $\gamma$. 
Furthermore, since $\check{b} = \emptyset$ and $|\check{b}| = |\check{\delta}|$, we have $\check{\delta} = \emptyset$ and therefore $M_2$ does not contain any job in $\delta$.
Besides the jobs $A\cup B \cup a \cup c \cup \gamma$, $M_2$ processes further jobs with total processing time $zD$. Therefore, all the jobs in $P$ are processed on $M_2$. We will now analyse where the jobs in $\gamma$ are scheduled.  
The only possibility where these jobs can be scheduled is the time between $a_i$ and $B_i$ for each $i \in \{1, \dots,z\}$ since at each other time the machine is occupied by other jobs. 
The processing time between the end of $a_i$ and the start of $B_i$ is exactly $\sched(B_i) - \sched(A_{i-1}) - p(A_{i-1}) - p(a_i)= D^5 + (3z-i)D^7$. The job in $\gamma$ with the largest processing time is the job $\gamma_1$ with $p(\gamma_1) = D^5+ (3z-1)D^7 -D$. 
This job only fits between $a_i$ and $B_1$.
Inductively we can show that $\gamma_i \in \gamma$ with $p(\gamma_i) = D^5+ (3z-i)D^7 -D$ has to be scheduled between $a_{i}$ and $B_{i}$ on $M_2$. 
Furthermore, since $p(\gamma_i) = D^5+ (3z-i)D^7 -D$ and the processing time between the end of $a_i$ and the start of $B_{i}$ is $D^5 + (3z-i)D^7$, there is exactly $D$ processing time left. 
These processing time has to be occupied by the jobs in $P$ since this schedule has no idle times. Therefore, we have for each $i \in \{1,\dots, z\}$ a disjunct subset $P_i \subseteq P$ containing jobs with processing times adding up to $D$. As a consequence $\mathcal{I}$ is a Yes-instance.
\end{proof}

\section{Hardness of Strip Packing}
\label{sec:HardnessStripPacking}

In the transformed schedule, all jobs are scheduled on contiguous machines. 
As a consequence, we have proven that this problem is strongly $NP$-complete even if we restrict the set of feasible solutions to those where all jobs are scheduled on continuous machines. 
We will now describe how this insight delivers a lower bound of $\frac{5}{4}$ for the best possible approximation ratio for pseudo-polynomial Strip Packing and in this way prove Theorem \ref{thm:StripPackingHardness}.

To show our hardness result for Strip Packing, let us consider the following instance. 
We define $W := (z+1)(D^2 + D^3 +D^8) + z(D^4 + D^5 +D^6) + z(7z + 1)D^7$ as the width of the considered strip, so it is the same as the considered makespan in the scheduling problem. 
For each job $j$ defined in the reduction above, we define an item i with $w(i) = p(j)$ and height $h(i) = q(j)$. 
Now, we can show analogously that if the 3-Partition instance is a Yes-instance there is a packing of height 4 (one example is the packing in Figure \ref{fig:packingStructure}) and if there is a packing with height $4$ then the 3-Partition instance has to be a Yes-instance. If the 3-Partition instance is a No-instance, the optimal packing has a height of at least 5 since the optimal height for this instance is integral. 
Therefore, we can not approximate Strip Packing in pseudo-polynomial time better than $\frac{5}{4}$.

\section{Conclusion}
In this paper, we positively answered the long standing open question whether $P4|size_j|C_{\max}$ is strongly $NP$-complete. Now, for each number of machines $m$ it is known whether the problem $Pm|size_j|C_{\max}$ is strongly $NP$-complete. Furthermore, we have improved the lower bound for pseudo-polynomial Strip Packing to $\frac{5}{4}$. Since the best known algorithm has an approximation ratio of $\frac{4}{3}$, this still leaves a gap between the lower bound and the best known algorithm. 
With the techniques used in this paper, a lower bound of $\frac{4}{3}$ for pseudo-polynomial Strip Packing can not be proven, since $P3|size_j|C_{\max}$ is solvable in pseudo-polynomial time and in the generated solutions all jobs are scheduled contiguously. Moreover, we believe that it is possible to find an algorithm with approximation ratio $\frac{5}{4} +\varepsilon$.
\bibliography{lowerBound}
%}

\end{document}